\documentclass{llncs}
\usepackage{setspace}
\usepackage{amssymb,amsmath}
\usepackage{url}
\usepackage{graphicx}
\usepackage{float}
\usepackage{algorithmic}
\usepackage{algorithm}
\usepackage{color}
\usepackage{amsmath}
\usepackage{amssymb}

\usepackage{epsfig}
\usepackage{epstopdf}
\usepackage{amsmath}
\usepackage{amssymb}

\newcommand{\myparagraph}[1]{\par\smallskip\par\noindent{\bf{}#1:~}}

\newcommand{\be}{\begin{equation}}
\newcommand{\ee}{\end{equation}}

\def \1{1 \!\! 1}

\newcommand{\cJ}{{\cal J}}
\newcommand{\cM}{{\cal M}}

\newcommand{\hepsilon}{h}

\newcommand{\remove}[1]{}
\newcommand{\comment}[1] {}

\newtheorem{observation}[definition]{Observation}
\newcommand{\cm}[1]{}
\def \ee   {\varepsilon}

\newlength{\tablength}
\newlength{\spacelength}
\newcommand{\tabstar}{\hspace*{\tablength}}
\newcommand{\spacestar}{\hspace*{\spacelength}}
\def\obeytabs{\catcode`\^^I=\active}
{\obeytabs\global\let^^I=\tabstar}
{\obeyspaces\global\let =\spacestar}
\newenvironment{display}{\begingroup\obeylines\obeyspaces\obeytabs}{\endgroup}
\newenvironment{prog}{\begin{display}\parskip0pt\sf}{\end{display}}

\begin{document}
%%%\begin{titlepage}

\title{Tighter Bounds for Makespan Minimization on Unrelated Machines
}

\author{
Dor Arad\inst{1}
\and Yael Mordechai\inst{1}
 \and Hadas Shachnai\inst{1}}

\institute{{Computer Science Department, Technion, Haifa 32000, Israel.
 \mbox{\tt \{dorarad@tx,myael@cs,hadas@cs\}.technion.ac.il}.}
}

\maketitle
%%%%%%%%%%%%%%%%%%%%%%%%%%%%%%%%%%%%%%%%%%%%%%%%%%%%%%%%%%%%%%%%%
\begin{abstract}
We consider the problem of scheduling $n$ jobs to minimize the makespan on $m$ unrelated
machines, where job $j$ requires time $p_{ij}$ if processed on machine $i$.
A classic algorithm
of Lenstra et al. \cite{LST90} yields the best known approximation ratio of $2$ for the problem.
Improving this bound has
been a prominent open problem for over two decades.

In this paper we obtain a tighter bound for
a wide subclass of instances which can be
identified efficiently.
Specifically, we define the {\em feasibility factor}
of a given instance as
the minimum fraction of machines on which each job can be processed.
We show that there is a polynomial-time algorithm that, given values
$L$ and $T$,
and an instance having a {\em sufficiently large} feasibility factor $\hepsilon \in (0,1]$,
either proves that no schedule of mean machine completion time $L$ and makespan $T$ exists, or else finds a
schedule of
makespan at most $T + L/\hepsilon$ which is smaller than $2T$ for a wide class of instances.

For the restricted version of the problem, where for each job $j$ and machine $i$, $p_{ij} \in \{p_j, \infty\}$,
we show that a simpler algorithm yields a better bound, thus improving for highly feasible instances
 the best known ratio of  $33/17 + \epsilon$, for any fixed $\epsilon >0$, due to
Svensson \cite{S12}.
\end{abstract}
%\bigskip
%
%\noindent
%{\bf Keywords:} generalized assignment problem, ad placement, group packing, approximation algorithms
%\bigskip

%\noindent {\bf AMS subject classifications:} 68W20, 68W25

%%%%%%%%%%%%%%%%%%%%%%%%%%%%%%%%%%%%%%%%%%%%%%%%%%%%%%%%%%%%%%%%%
%%%%%%%%%%%%%%%%%%%%%%%%%%%%%%%%%%%%%%%%%%%%%%%%%%%%%%%%%%%%%%%%%

%%%%%%%%%%%%%%%%%%%%%%%%%%%%%%%%%%%%%%%%%%%%%%%%%%%%%%%%%%%%%%%%%
\section{Introduction}
\label{sec:intro}

In the problem of scheduling on unrelated parallel machines, we are given
%The problem of makespan minimization on unrelated parallel machines is one of the most  prominent and  important
% problems in the
%area of machine scheduling. In this model, we are given
a set ${\cJ}$ of jobs to be processed without interruption on a set ${\cM}$ of unrelated
machines, where the time a machine $i \in {\cM}$ needs to process a job $j \in {\cJ}$ is specified by
a machine and job dependent processing time $p_{ij} \geq 0$.
When considering a scheduling
problem, the most common and perhaps most natural objective function is makespan
minimization. This is the problem of finding
a schedule (or, an assignment),
$\alpha:{\cJ} \rightarrow {\cM}$,
so as to minimize the time $max_{i \in {\cM}} \sum_{j \in \alpha^{-1}(i)} p_{ij} $ required to process all
jobs.
A classic result in scheduling theory is the Lenstra-Shmoys-Tardos $2$-approximation
algorithm for this fundamental problem \cite{LST90}. Their approach is based on several nice
structural properties of the extreme point solutions of a natural linear program and
has become a textbook example of such techniques (see, e.g., \cite{V01}).  Complementing
their positive result, they also proved that the problem is NP-hard to approximate
within a factor less than $3/2$, even in the restricted case (also known as the {\em restricted assignment
problem}), where $p_{ij} \in \{p_j, \infty \}$.
%(i.e., when
%job $j$ has processing time $p_j$ or $\infty$ for each machine).
Despite being a prominent open problem in scheduling theory, there has been
very little progress on either the upper or lower bound since the publication of \cite{LST90}
over two decades ago, with the exception being the
recent beautiful result of Svensson \cite{S12}, showing
for the restricted case
an upper bound of $33/17+\epsilon$, for an arbitrarily small constant $\epsilon >0$.

In this paper we show that these
%the above
best known bounds can
% tighter bounds can be obtained for
be tightened for
a wide subclass of instances,
%, of both the general case and the restricted version
%of scheduling on unrelated machines,
which can be identified efficiently.
Specifically, we define the {\em feasibility factor}
of a given instance as
the minimum fraction of machines on which each job can be processed.
%For the general problem of makespen minimization on unrelated machines,
We show that there is a polynomial-time algorithm that, given values
$L$ and $T$,
and an instance having a {\em sufficiently large} feasibility factor $\hepsilon \in (0,1]$,
either proves that no schedule of mean machine completion time $L$ and makespan $T$ exists, or else finds a
schedule of
makespan at most $T + L/\hepsilon < 2T$.
For the restricted assignment problem,
we show that a simpler algorithm yields a better bound, thus
enabling to improve
% improving
 for highly feasible instances
 the best known ratio of $33/17 + \epsilon$
%due to Svensson
of \cite{S12}.

We note that the feasibility factor $\hepsilon$ of a given instance has been used before, both
for improving upper bounds (see, e.g., \cite{ETKS08}) and for showing the hardness
of certain subclasses of instances \cite{JW11}. However, these previous studies focus on specific values
of $\hepsilon$ (see below).
Our study takes a different approach in
exploring the decrease in the makespan that can be
achieved, by identifying instances in which $\hepsilon$ is sufficiently large (see Section \ref{sec:results}).

\subsection{Prior Work}
Minimizing the makespan on unrelated parallel machines
has been extensively studied for almost four decades.
Lenstra et al. \cite{LST90}
introduced an LP-based polynomial time $2$-approximation algorithm for the problem.
% The algorithm of \cite{LST90} computes
%first an optimal fractional solution for the linear programming relaxation of the problem and then uses rounding to obtain an
%approximate schedule for the discrete problem.
They also
proved that unless $P=NP$, there is
no polynomial time approximation algorithm
with approximation factor better than $\frac{3}{2}$.
Gairing et al. \cite{GMW07} presented a combinatorial $2$-approximation algorithm for the problem.
%makespan on unrelated machines. Their algorithm
%improves the running time of the LP-based algorithm of \cite{LST90}.
Shchepin and Vakhania \cite{SV05} showed that the rounding technique used in \cite{LST90}
can be modified to derive an improved ratio
%presented a rounding technique achieving
a factor of
 $2-\frac{1}{m}$.
%, thus
%slightly improving upon the approximation ratio of 2.

Shmoys and Tardos \cite{ST93} further developed the technique of \cite{LST90} to obtain an  approximation ratio
of $2$
 for the {\em generalized assignment problem (GAP)}, defined as
follows. We are given a set of jobs, ${\cJ}$, and a set of unrelated machines, ${\cM}$.
Each job is to be processed by exactly one machine;
processing job $j$ on machine $i$ requires time $p_{ij}>0$ and
incurs a cost of $c_{ij}>0$.
Each machine $i$ is available for $T_i$ time units, and the objective is to minimize
the total cost incurred.
%The objective is to find a schedule
%that minimizes the makespan and the total cost incurred.
The paper \cite{ST93} presents a polynomial time algorithm that, given values $C$
and $T$, finds a schedule of cost at most $C$ and makespan at most
$2T$, if a schedule of cost $C$ and makespan $T$ exists. This is
 the best known result to date for GAP.

%The \emph{restricted assignment problem} is a special case of scheduling
%on unrelated machines, in which the processing time of each job $j$
%on machine $i$ is $p_{ij}\in\left\{ p_{j},\infty\right\} $, for
%$1\leq j\leq n,$ $1\leq i\leq m$. This problem is also hard to approximate
%within factor $\frac{3}{2}$.
For the restricted assignment problem,
Gairing et al. \cite{GLMM04} presented a combinatorial $2-\frac{1}{p_{max}}$-approximation algorithm
based on flow techniques, where $p_{max} = \max_j p_j$ is the maximum processing time of any job in the given instance.
The best known approximation ratio is
$\frac{33}{17}+\epsilon=1.9412+\epsilon$, due to Svensson \cite{S12}.

Interestingly, the feasibly factor of a given instance served as
a key component in deriving
two fundamental results
for the restricted assignment problem.
Ebenlendr et al. \cite{ETKS08} showed that the
subclass of instances for which
$\hepsilon = 2/m$ admits an approximation factor of $1.75$, thus improving for this subclass the general
bound of $2$. The same
problem,
also called {\em unrelated graph balancing},
was studied
by Verschae and Wiese \cite{JW11}. They showed that, in fact, this surprisingly simple subclass of instances
constitutes the core difficulty for the {\em linear programming} formulation of the
problem, often used as a first step in obtaining approximate solutions. Specifically, they
showed that
already for this basic setting, the strongest
known LP-formulation, namely,  the configuration-LP, has an integrality gap
of $2$.

\subsection{Our Contribution}
\label{sec:results}
In this paper we improve the best known bounds for makespan minimization
on unrelated parallel machine, for a wide subclass of instances possessing high
feasibility factor.
%by considering the feasibility factor of a given instance.
In particular, in Section \ref{sec:general}, we
 show that there is a polynomial-time algorithm that, given values
$0 <L < T$,
and an instance $I$ having a feasibility factor $L/T \leq \hepsilon$,
either proves that no schedule of mean machine completion time $L$ and makespan $T$ exists, or else finds a
schedule of
makespan at most $T + L/\hepsilon < 2T$.

For the restricted assignment problem, let $L= \frac{\sum_{j=1}^n p_j}{m}$ be the mean machine completion time of {\em any} schedule.
Then we show that 
there is an $O(m^2n)$ time algorithm that, given 
an instance $I$ whose feasibility factor satisfies $L/p_{max} = q < \hepsilon$,
finds a schedule of
makespan at most $p_{max} + L/\hepsilon  < (1+q/h)OPT$,
where $OPT$ is the makespan of an optimal schedule.
\comment{
We show (in Section \ref{sec:restricted}) that
for any instance
%in which the
with a feasibility factor satisfying $L/p_{max} = q < \hepsilon$,
there is an $O(m^2 n)$ time algorithm that outputs a makespan at most $p_{max} + L/\hepsilon  < (1+q)OPT$,
where $OPT$ is the makespan of an optimal schedule.
%For any instance where
}
Thus, for $q \leq \frac{16}{17} \hepsilon$, this improves the bound of $33/17 + \epsilon$ of \cite{S12}.

\myparagraph{Techniques} Our algorithms rely heavily on the fact that the given instances are highly feasible, and
thus, the schedules can be better balanced to decrease the latest completion time of any job.
Our algorithm for general instances first uses as a subroutine 
%identifies a highly feasible instance, by using as a subroutine
an algorithm of \cite{ST93}, thus also identifying highly feasible instances. It then applies on the resulting schedule a {\em balancing} phase.
Using some nice properties of this schedule, the algorithm moves long jobs
from overloaded to underloaded machines, while decreasing the makespan of the schedule.

Our bound for the restricted assignment problem builds on a result of Gairing at al. \cite{GLMM04}, who gave an
algorithm based on flow techniques for general instances of the problem. Their algorithm starts with a
schedule satisfying certain properties and gradually improves the makespan until it is guaranteed to yield a ratio of $2$ to the
optimal. The main idea is to use some parameters for partitioning the machines into three sets: overloaded,
underloaded, and all the remaining machines. The makespan is improved by moving jobs from overloaded to underloaded machines on augmenting path
in the corresponding flow network.
We adopt this approach and show that by a {\em good} selection of the parameters defining the three machine sets, we can obtain
the desired makespan. Consequently, our algorithm is simpler and has better running time than the algorithm of
\cite{GLMM04} (see Section \ref{sec:restricted}).

\section{Preliminaries}
\label{sec:prel}

An \emph{assignment} of jobs to machines is given by a function
$\alpha:{\cJ} \rightarrow {\cM}$. Thus, $\alpha(j)=i$ if job $j$ is
assigned to machine $i$. For any assignment $\alpha$, the \emph{load}
$\delta_{i},$ on machine $i$, given a matrix of processing times $\mathbf{P}$,
is the sum of processing times for the jobs that were assigned to
machine $i$, thus $\delta_{i}(\mathbf{P},\alpha)=\sum_{j\in {\cJ}:\alpha(j)=i}p_{ij}$.
The \emph{makespan} of an assignment $\alpha$ is the maximum load on any machine.
Also, the \emph{average machine load} (or, mean machine completion time) is given by $L=\frac{\sum_{i\in {\cM}}\delta_{i}(\mathbf{P},\alpha)}{m}$.

Given the matrix $\mathbf{P}$ and the value $T>0$, we say that a machine
$i$ is \emph{legal }for job $j$ if $p_{ij}\leq T$.
Thus, any job $j \in {\cJ}$ can be assigned to at least $\hepsilon m$ machines in ${\cM}$.
The \emph{feasibility factor } of $\mathbf{P}$ is $\hepsilon(T) =\frac{min_{j\in {\cJ}}\left|\left\{ i\in {\cM} \,:i\,\mbox{is legal for \ensuremath{j}}\right\} \right|}{m}$.

Given an assignment $\alpha$ for a matrix $\mathbf{P}$, and a constant $\gamma\geq1$, we denote by $Bad(\mathbf{P},\alpha,\gamma)$
the set of machines that complete processing after time $T+\gamma\cdot L$, by
$Good(\mathbf{P},\alpha,\gamma)$ the set of machines that complete
by time $\gamma\cdot L$, and by $Good_{j}(\mathbf{P},\alpha,\gamma)$
the set of machines from $Good(\mathbf{P},\alpha,\gamma)$ that are
legal for job $j$. Also we denote by $j_{max}^{i}(\mathbf{P},\alpha)=argmax\left\{ p_{ij}:\,\alpha(j)=i\right\} $
the largest job on machine $i\in Bad(\mathbf{P},\alpha,\gamma)$.
We denote by $G_{\alpha,\gamma}(\mathbf{P},\alpha)$ the bipartite
graph $\left(\left(Bad(\mathbf{P},\alpha,\gamma),Good(\mathbf{P},\alpha,\gamma)\right),E(\mathbf{P},\alpha)\right)$
where $E(\mathbf{P},\alpha)$ consists of all edges $(i,i^{'})$ where
$i^{'}$ is a legal machine for $j_{max}^{i}(\mathbf{P},\alpha)$.
We say that machine $i$ is\emph{ good }if $i\in Good(\mathbf{P},\alpha,\gamma)$.
We say that machine $i$ is\emph{ bad }if $i\in Bad(\mathbf{P},\alpha,\gamma)$.
Machine $i$ is\emph{ good for job $j$ }if $i\in Good_{j}(\mathbf{P},\alpha,\gamma)$.
We omit $\mathbf{P}$ in the notation if it is clear
from the context.

\section{Approximation Algorithm for General Instances}
\label{sec:general}
The problem of scheduling on unrelated machines can be viewed as a
special case of the \emph{generalized assignment problem} in which $c_{ij}=0$ for
all $1\leq i\leq m$ and $1\leq j\leq n$. The best known result for
the generalized assignment problem is due to \cite{ST93}. They presented
a polynomial time algorithm that, given values $C$ and $T$, finds
a schedule of cost at most $C$ and makespan at most $2T$, if a schedule
of cost $C$ and makespan $T$ exists. This implies also the best
result for scheduling on unrelated machines - a schedule of makespan
at most twice the optimum.
In our result we use the rounding technique of \cite{ST93} with the costs being the processing times. This allows us to bound the average machine's completion time of the resulting assignment. A bound that is essential for our result.
In the next Section we give an overview of the rounding technique of \cite{ST93}.

\subsection{Overview of the Algorithm of Shmoys and Tardos}
We describe below the technique
used in \cite{ST93} for solving the generalized assignment problem. Let $\mathbf{P}$ and
$\mathbf{C}$ denote the matrix of processing times and the matrix of costs, and let $T$  and $C$
 be fixed positive integers. Let the indicator variables $x_{ij}$, $i=1,2,...,m$ , $j=1,2,...,n$ denote whether job $j$
is assigned to machine $i$. Then the linear programming relaxation
of the problem is as follows:

%\begin{figure}
\begin{center}
%\fbox{
\begin{minipage}{0.95\textwidth}
\begin{eqnarray*}
LP(\mathbf{P},\mathbf{C},T,C): & \displaystyle{\sum_{i=1}^{m}\sum_{j=1}^{n}c_{ij}x_{ij}\leq C} \hbox{~~~~~~~~~~~~~~~~~~~~~~~~~~~~~~~~ }\\
& \displaystyle{\sum_{i=1}^{m}x_{ij}=1}, \hbox{~~~~~~~~~~ for \ensuremath{j=1,...,n} ~~~~~~~~ }  \\
& \displaystyle{\sum_{j=1}^{n}p_{ij}x_{ij}\leq T}, \hbox{~~~~~ for \ensuremath{i=1,...,m} ~~~~~~~~ } \\
& ~~~~~x_{ij}\geq0, \hbox{~~~~~~~~~~~~~~~~ for  \ensuremath{i=1,...m\,,j=1,...,n} } \\
& ~~~~x_{ij}=0, ~if ~ p_{ij}>T, \hbox{ for \ensuremath{i=1,...m\,,j=1,...,n}} \\
&&
\end{eqnarray*}
\end{minipage}
%}
\end{center}
%\end{figure}

\comment{
\begin{eqnarray*}
LP(\mathbf{P},T,C): & \sum_{i=1}^{m}\sum_{j=1}^{n}c_{ij}x_{ij}\leq C\\
 & \sum_{i=1}^{m}x_{ij}=1 & \mbox{for \ensuremath{j=1,...,n}}\\
 & \sum_{j=1}^{n}p_{ij}x_{ij}\leq T & \mbox{for \ensuremath{i=1,...,m}}\\
 & x_{ij}\geq0 & \mbox{for  \ensuremath{i=1,...m\,,j=1,...,n}}\\
 & x_{ij}=0 & \mbox{if \ensuremath{p_{ij}>T\,,i=1,...m\,,j=1,...,n}}
\end{eqnarray*}
}

Let $x_{ij}$, $i=1,2,...,m$ , $j=1,2,...,n$ be the fractional solution, and let $k_{i}=\left\lceil \Sigma_{j=1}^{n}x_{ij}\right\rceil $. Each machine is partitioned into $k_{i}$ sub-machines $v_{i,s}$,
$s=1,...,k_{i}]$. The rounding is done by finding
a minimum-cost perfect matching between all jobs and all sub-machines.
Formally, a bipartite graph $B(W,V,E)$ is constructed, with $W=\left\{ w_{j}:j=1,...,n\right\} $
the job nodes, $V=\left\{ v_{is}:i=1,..,m,\, s=1,...,k_{i}\right\} $
the sub-machines nodes. Each node $v_{is}$ can be viewed as a bin
of volume 1, and we add an edge $(w_{j},v_{is})$ with cost
$c_{ij}$ iff a positive fraction of $x_{ij}$ is packed in the bin $v_{is}$. For every machine $i=1,..,m$ the jobs are sorted in non-increasing order of their processing time on $i$. Then, the bins $v_{i1},...,v_{ik_{i}}$ are packed one by one, with the values
$x_{ij}>0$ by the order of the jobs. While $v_{is}$ is not totally
packed, we continue packing the $x_{i,j}$s such that if $x_{ij}$
fits $v_{is}$ it is packed to $v_{is}$, else, only a fraction $\beta$
of $x_{ij}$ is packed to $v_{is}$, consuming all the remaining volume
of $v_{is},$; the remaining part of $(1-\beta)x_{ij}$ is packed in $v_{i,s+1}$.
Then the rounding is done by taking a minimum-cost integer matching
$M$, that matches all job nodes, and for each edge $(w_{j},v_{is})\in M$
, set $x_{ij}=1,$ i.e., schedule job $j$ on machine $i$.

The resulting schedule has the following nice property, that is used below for deriving our result.

\begin{lemma}\cite{ST93}
\label{lemma:small-jobs_sum}
Let $\alpha$ be the assignment obtained by the algorithm and let $j_{max}^{i}=max\left\{ p_{ij}:\alpha(j)=i\right\}$ the longest jobs that was assigned to $i$. Then,
for all $1\leq i\leq m$, $\sum_{j:\alpha(j)=i,\, j\neq j_{max}^{i}}p_{ij}\leq T$.
\end{lemma}

\subsection{Approximation Algorithm}

Consider the special case of the generalized assignment problem in which the costs satisfy $c_{ij}=p_{ij}$,
for all $i=1,...,m$ and $j=1,...,n$. For any instance $\mathbf{P}$
and the constants $T$ and $L\leq T$, integral solutions to the
linear program, $LP(\mathbf{P},\mathbf{P},T,L)$, are in one-to-one
correspondence with schedules of makespan at most $T$, and average
machine load $L$.

Thus, the result of \cite{ST93} guarantees that if $LP(\mathbf{P},\mathbf{P},T,L)$
has a feasible solution, then there exists a schedule that has makespan
at most $2T$ and average machine load $L$.
Our main result is the following.

\begin{theorem}
\label{thm:approx_ratio}
Let $T$ and $L\leq T$ be some fixed positive values for a given instance $\mathbf{P}$ of the scheduling problem, let $\hepsilon\left(T\right)=\hepsilon\in\left(0,1\right] $ be the feasibility factor of $\mathbf{P}$.If $LP(\mathbf{P},T,L)$ has a feasible solution, then there is an algorithm that achieves a makespan of $min\left\{ T+\frac{L}{\hepsilon},2T\right\}$.
\end{theorem}

We prove the theorem by describing an algorithm that converts
a feasible solution for $LP(\mathbf{P},\mathbf{P},T,L)$ to the desired schedule.
The first step of the algorithm is to apply the rounding technique
of \cite{ST93} to obtain an assignment $\alpha$ that admits
makespan at most $2T$ and average machine load at most $L$. Next,
the algorithm fixes the assignment to achieve a makespan of at most
$T+\frac{L}{\hepsilon}$. This is done by transferring
the largest job from each machine $i\in Bad(\mathbf{P},\alpha,\frac{1}{\hepsilon})$
to a machine $i^{'}\in Good(\mathbf{P},\alpha,\frac{1}{\hepsilon})$.
To prove that all transfers are possible, we show that
there exists a perfect matching between the bad machines and the good
machines. Formally, we prove that there exists a perfect matching
in $G_{\alpha,\frac{1}{\hepsilon}}(\mathbf{P},\alpha)$. We first prove the following lemmas.

\begin{lemma}
\label{enough_good_machines}
Let $\alpha$ be an assignment that admits a makespan of at most $2T$
and average machine load \emph{ $L\leq T$. } Let $\gamma\geq1$
and assume $\left|Bad(\alpha,\gamma)\right|=k$. Then
\begin{enumerate}
\item $k<\frac{m}{\gamma+1}$.
\item $\left|Good(\alpha,\gamma)\right|>\left(1-\frac{1}{\gamma}\right)\cdot m+\frac{k}{\gamma}\cdot\frac{T}{L}$.\end{enumerate}

\end{lemma}

\begin{proof}
Each machine $i\in Bad(\alpha,\gamma)$ has load greater than $T+\gamma\cdot L$,
therefore $\sum_{i\in M}\delta_{i}(\alpha)>k\cdot(T+\gamma\cdot L)$.
\begin{enumerate}
\item Assume that $k\geq\frac{m}{\gamma+1}$, then

\[
\begin{array}{ll}
\sum_{i\in M}\delta_{i}(\alpha) & >  k(T+\gamma L)\\
 & \geq \frac{m}{\gamma+1}(T+\gamma L)\\
 & = \frac{m}{\gamma+1}(T+\gamma L)\\
& {\geq}
 \frac{m}{\gamma+1}(L+\gamma L)\\
 & =  m\cdot L
\end{array}
\]

The last inequality follows from the fact that $T \geq L$.
Hence, the average machine load is greater than $L$, a contradiction.
It follows that $k<\frac{m}{\gamma+1}$.

\item Let $\left|Good(\alpha,\gamma)\right|=l$.
Then, there are $m-k-l$ machines having loads greater than $\gamma L$.
Assume that $l\leq\left(1-\frac{1}{\gamma}\right)m+\frac{k}{\gamma}\cdot\frac{T}{L}$, then

\[
\begin{array}{ll}
\sum_{i\in M}\delta_{i}(\alpha) & >  k(T+\gamma\cdot L)+\left(m-l-k\right)\gamma L\\
 & = kT+\left(m-l\right)\gamma L\\
 & \geq  kT+\left(m-\left(1-\frac{1}{\gamma}\right)m+\frac{k}{\gamma}\cdot\frac{T}{L}\right)\gamma L\\
 & \geq kT+\left(\frac{m}{\gamma}+\frac{k}{\gamma}\cdot\frac{T}{L}\right)\gamma L\\
 & =  kT+\left(mL+kT\right)\\
 & \geq mL
\end{array}
\]

Hence, the average machine load is greater than $L$, a contradiction.
It follows that $\left|Good(\alpha,\gamma)\right|\geq\left(1-\frac{1}{\gamma}\right)\cdot m+\frac{k}{\gamma}
\cdot\frac{T}{L}$. \hfill \qed
\end{enumerate}
\end{proof}

\begin{lemma}
Let $\mathbf{P}$ be an instance of the scheduling problem. Let $\alpha$
be an assignment for $\mathbf{P}$ that admits a makespan of at most
$2T$ and average machine load \emph{ $L\leq T$}. If $\hepsilon\geq\frac{L}{T}$ then for every subset $A\subseteq Bad(\alpha,\frac{1}{\hepsilon})$
, $\left|N\left(A\right)\right|\geq\left|A\right|$, where $N\left(A\right)$
is the set of neighbors of $A$ in $G_{\alpha,\gamma}$.\end{lemma}
\begin{proof}
Let $\left|Bad(\alpha,\frac{1}{\hepsilon})\right|=k$. Since the number
of illegal machines for any job $j$ is at most $(1-\hepsilon)m$, the number of good machines for job $j$ is at least the number
of good machines minus its illegal machines (the worst case where
all illegal machines for job $j$ form a subset of $Good(\alpha,\frac{1}{\hepsilon})$).
Together with Lemma \ref{enough_good_machines} we have

\[
\begin{array}{ll}
\left|Good_{j}(\alpha,\frac{1}{\hepsilon})\right| & \geq  \left|Good(\alpha,\frac{1}{\hepsilon})\right|-(1-\hepsilon)m\\
 & >  \left(1-\frac{1}{\gamma}\right)\cdot m+\frac{k}{\left(\frac{1}{\hepsilon}\right)}\cdot\frac{T}{L}-(1-\hepsilon)m\\
  & = \hepsilon\cdot k\frac{T}{L}\\
  & \geq k
\end{array}
\]

The last inequality follows from the fact that $\hepsilon\ge\frac{L}{T}$.
Now, let $A\subseteq Bad(\alpha,\frac{1}{\hepsilon})$. Then $\left|A\right|\leq\left|Bad(\alpha,\frac{1}{\hepsilon})\right|=k$.
Recall that the set of neighbors of $A$ is the set of machines that
are good for all the jobs $j_{max}^{i}$, $i\in A$, i.e., $N\left(A\right)=\cup_{i\in A}Good_{j_{max}^{i}}(\alpha,\frac{1}{\hepsilon})\subseteq Good(\alpha,\frac{1}{\hepsilon})$.
Obviously $\left|N\left(A\right)\right|=\left|\cup_{i\in A}Good_{j_{max}^{i}}(\alpha,\frac{1}{\hepsilon})\right|\geq\left|Good_{j_{max}^{i}}(\alpha,\frac{1}{\hepsilon})\right|$
for some $i\in A$. It follows from the above that $\left|N\left(A\right)\right|\geq k$.

Since $\left|A\right|\leq k$ we have that $\left|N\left(A\right)\right|\geq\left|A\right|$.
\hfill \qed
\end{proof}

By Hall's Theorem \cite{H35}, there exist a perfect matching in
$G_{\alpha,\frac{1}{\hepsilon}}$ {\em iff}  for every $A\subseteq Bad(\alpha,\frac{1}{\hepsilon})$,
$\left|N\left(A\right)\right|\geq\left|A\right|.$ Thus, we have
\begin{corollary}
\label{lemma:hall}
There exists a perfect matching in $G_{\alpha,\frac{1}{\hepsilon}}$.
\end{corollary}

By the above discussion, we can modify the assignment $\alpha$,
output by the algorithm of \cite{ST93}, by finding a perfect
matching in $G_{\alpha,\frac{1}{\hepsilon}}$ and then transferring jobs
from bad machines to their matching good machines. We describe this formally in algorithm $A_{UM}(\mathbf{P},\gamma,T,L)$.

\begin{algorithm}[H]
\caption{$A_{UM}(\mathbf{P},\gamma,T,L)$}
\begin{enumerate}
\item Solve the linear relaxation $LP(\mathbf{P},T,L)$.
\item Round the solution to obtain an integral assignment $\alpha$ using a rounding technique as
given in \cite{ST93}. \label{alg:step2}
\item If the feasibility factor $\hepsilon$ of $\mathbf{P}$ satisfies $\hepsilon\leq\frac{L}{T}$, then return the assignment $\alpha$.  \label{alg:step3}
\item Otherwise, construct the bipartite graph $G_{\alpha,\frac{1}{\hepsilon}}$
and find a perfect matching of size $\left|Bad(\mathbf{P},\alpha)\right|$.
\item Obtain a resulting assignment $\beta$ out of $\alpha$ by transferring
the longest job, $j_{max}^{i}$ from each machine $i\in Bad(\mathbf{P},\alpha)$,
to its matching machine $i^{'}\in Good(\mathbf{P},\alpha)$.
\item Return the new assignment $\beta$. \end{enumerate}
\end{algorithm}

\noindent
{\bf Proof of Theorem \ref{thm:approx_ratio}.}
We show that the assignment output by Algorithm $A_{UM}$
satisfies the statement of the theorem. Consider an instance $\mathbf{P}$.
By \cite{ST93}, if $LP(\mathbf{P},\mathbf{P},T,L)$ has a feasible solution
for $C\leq T$ then Step \ref{alg:step2}. is guaranteed to generate a schedule of makespan
at most $2T$ and average machine load $L$. Let $\alpha$ be the
resulting assignment.

Let $\hepsilon$ be the feasibility factor of $\mathbf{P}$. If $\hepsilon\leq\frac{C}{T}$
then we output $\alpha$ at Step \ref{alg:step3}, and indeed, we cannot guarantee a
makespan lower than $2T$ in this case.
Otherwise, by Corollary \ref{lemma:hall}, there exists a perfect matching in $G_{\alpha,\frac{1}{\hepsilon}}$.
Let $\left|Bad(\mathbf{P},\alpha)\right|=k $ and let $M=\left\{ (i_{b_{1}},i_{g_{1}}),...,(i_{b_{k}},i_{g_{k}})\right\} $
be a perfect matching in $G_{\alpha,\frac{1}{\hepsilon}}.$

By Lemma \ref{lemma:small-jobs_sum}, for any
machine $i=1,...,m$ the sum of processing times of all the jobs $j$
such that $j\in\left\{ j:\alpha(j)=i\right\} \setminus\left\{ j_{max}^{i}\right\} $
is at most $T$. Therefore, transferring the largest job $j_{max}^{i}$
to another machine guarantees that the new load of $i$ is at
most $T$.

As for the good machines, if $i$ is a good machine for job $j$ then
$p_{ij}\leq T$, then transferring $j$ to $i$ will increase the
load of $i$ at most by $T$. Since the load of a good machine is
at most $\frac{L}{\hepsilon}$, we have that after such job transfer
the load is at most $T+\frac{L}{\hepsilon}$.

We note that each pair $(i_{b_{s}},i_{g_{s}})\in M$ is a matching
between $i_{b_{s}}\in Bad$ and $i_{g_{s}}\in Good_{j_{max}^{i_{b_{s}}}}$
and therefore transferring the largest job $j_{max}^{i_{b_{s}}}$
from $i_{b_{s}}$ to $i_{g_{s}}$ guarantees that the resulting load on $i_{b_{s}}$
is at most $T$, and  the load on $i_{g_{s}}$ is at
most $T+\frac{L}{\hepsilon}$.

Thus, by performing the large-jobs transfers for all pairs $(i_{b_{s}},i_{g_{s}})\in M$, $s=1,...,k$,
we guarantee that each machine has load at most $T+\frac{L}{\hepsilon}$.
 \hfill $\qed$

\section{A Better Bound for the Restricted Assignment Problem}
\label{sec:restricted}
In this section we consider the restricted version of our problem, where $p_{ij}\in\left\{ p_{j},\infty\right\}$, for each job $j=1,...,n$, and each machine $i=1,..,m$.
For this variant, we show that improved approximation ratio can be achieved by a combinatorial algorithm.
In particular, applying a technique of Gairing et al. \cite{GLMM04}, we show that, by identifying highly feasible instances, we obtain an algorithm
which improves the $2$-approximation ratio guaranteed in \cite{GLMM04}, and also has better running time.
In the next section we give an overview of the algorithm of \cite{GLMM04}.

\subsection{Overview of the Algorithm of Gairing et al.}
\label{sec:Gairing overview}

We describe below an algorithm, called Unsplittable-Blocking-Flow, introduced in \cite{GLMM04}.
Let $I$ be an instance for the restricted assignment problem. Also, let
$w$ be a fixed positive integer and $\Delta$ a parameter (to be determined).
A $w$-feasible assignment $\alpha$ is an assignment with the property that each job $j$ is assigned to a machine $i$ where $p_{ij}\leq w$.

Let $\alpha$ be a $w$-feasible assignment. $G_\alpha(w)=(W\cup V, E_\alpha (w))$ is a directed bipartite graph where
 $W=\left\{ w_{j}:j=1,...,n\right\}$ consists of the job nodes, and $V=\left\{ v_{i}:i=1,..,m\right\}$ consists
of machine nodes. For any job node $j$ and any machine node $i$, if $\alpha(j)=i$ there is an arc in $E_\alpha(w)$ from $i$ to $j$; if $\alpha(j)\neq i$ and $j$ is feasible on machine $i$, i.e., $p_{ij}\leq w $, then there is an arc from $j$ to $i$.

Given a $w$-feasible assignment $\alpha$, the algorithm of  \cite{GLMM04} partitions the set of machines to
three subsets: ${\cal M}^+$ (overloaded), ${\cal M}^-$ (underloaded), and ${\cal M}^{0}$ (all the remaining machines).
Thus, 
${\cal M}={\cal M}^+\cup {\cal M}^-\cup {\cal M}^0$.
Given an assignment $\alpha$,
a machine $i\in {\cal M}^{+}$ is overloaded if the load on $i$ is at least
$ w+\Delta+1 $. A machine $i\in {\cal M}^{-}$ is underloaded if the load on $i$ is at most $\Delta $.
 The remaining machines, which are neither overloaded nor underloaded,
form the set ${\cal M}^{0}={\cal M}\smallsetminus\left({\cal M}^{-}\bigcup {\cal M}^{+}\right)$.

The algorithm Unsplittable-Blocking-Flow$({\cal M},\alpha, \Delta, w)$  starts with an initial $w$-feasible assignment of jobs to machines and iteratively improves the makespan until it obtains an assignment with makespan of $w+\Delta$,
or declares that an assignment of makespan $\Delta$ does not exist.
In each iteration, the algorithm finds an augmenting path from an overloaded
to an underloaded machine, and pushes jobs along this path, by performing a series of job reassignments 
between machines on that path. This results in balancing the load over the machines, i.e., reducing the
load of the source that is an overloaded machine, and increasing the load of the destination that is an underloaded machine, while preserving the load of all other machines.
Unsplittable-Blocking-Flow terminates after $O(mS)$ steps, where
$S=\Sigma_{j=1}^{n}|\left\{ i:p_{ij}<\infty\right\} |$.
For short, we call this algorithm below $\cal UBF$.

Algorithm $\cal UBF$ (with $w=p_{max}$), combined with a binary search over the possible range of values for $\Delta$, 
can be used to obtain the approximation ratio of $2-\frac{1}{p_{max}}$.
The running time of the approximation algorithm is then factored by a value that is logarithmic in the size of the range in which we search for $\Delta$.
Thus,
the  algorithm of \cite{GLMM04} computes an assignment having makespan within a factor of $2-\frac{1}{p_{max}}$ from the optimal in time $O(mSlogW)$, where $W=\Sigma_{j=1}^{n}p_{j}$.

\subsection{Approximation Algorithm}
\label{sec:Approximation algorithm}

Let $\cal I$ be an instance of the restricted assignment problem.
The feasibility factor of $\cal I$ is exactly $\hepsilon=\frac{min_{j}\left|\left\{ i:p_{ij}<\infty\right\} \right|}{m}$.
Consider the bipartite graph $G_\alpha(w)$ constructed by algorithm $\cal UBF$.
%We note that, in the restricted case, the following property holds.
%
%\begin{observation}
%label{obs:load}
%Let $\alpha$ be a $w$-feasible assignment of average load $L$. Let $Q$ be a path in $G_\alpha (w)$ with its origin and endpoint at machine nodes.
%Then reassigning jobs along $Q$ results in a new $w$-feasible assignment of average load exactly $L$. Furthermore, if  $w\geq max_j{p_j}$ then any complete assignment of the jobs is $w$-feasible,
% and $L=\frac {1}{m}\Sigma_{j=1}^{n}p_{j}$.
%\end{observation}
\comment{
\begin{observation}
\label{obs:load}
Let $w\geq max_j{p_j}$. Then any non partial assignment is $w$-feasible and the average machine load of any $w$-feasible assignment equals $\frac {1}{m}\Sigma_{j=1}^{n}p_{j}$.
\end{observation}
}
%Therefore, the optimal average machine load of any feasible restricted instance is $\frac {1}{m}\Sigma_{j=1}^{n}p_{j}$.

\comment{
\begin{observation}
The feasibility factor is exactly $\epsilon=\frac{min_{j}\left|\left\{ i:p_{ij}<\infty\right\} \right|}{m}$.
\end{observation}
}

Let $L=\frac{1}{m}\sum_{j=1}^{n}p_{j}$ be the  average machine load
of any schedule for  an instance of the restricted assignment problem.
%We modify the definition underloaded and overloaded 
%machines when considering only $w\geq p_{max}$ and $\Delta=\frac{L}{\hepsilon}$ and general processing times.
%Therefore we get the next definition of underloaded and overloaded machines.
In the following we define the three machine sets for the algorithms.
\begin{definition}
\label{def:machinesdef}
Let $\alpha$ be a $w$-feasible assignment  of a given instance  $\cal I$  of the restricted assignment problem, 
with feasibility factor $\hepsilon $. Then,
%let $w\geq p_{max}$ and $L=\frac{1}{m}\Sigma_{j=1}^{n}p_{j}$  and let $\alpha$ be a $w$-feasible assignment.

${\cal M}^{-}\left(\alpha\right)=\left\{ i:\delta_{i}({\cal I},\alpha)\leq \frac{L}{\hepsilon} \right\} $

${\cal M}^{0}\left(\alpha\right)=\left\{ i:\frac{L}{\hepsilon} <\delta_{i}({\cal I},\alpha)\leq w+\frac{L}{\hepsilon}\right\} $

${\cal M}^{+}\left(\alpha\right)=\left\{ i:\delta_{i}({\cal I},\alpha)> w+\frac{L}{\hepsilon}\right\} $
\end{definition}

Using this partition of machines into ${\cal M}^{0}$, ${\cal M}^{-}$ and
${\cal M}^{+}$, we apply the algorithm $\cal {UBF}$$({\cal M},\alpha,\frac{L}{\hepsilon}, p_{max})$ of \cite{GLMM04}.

Throughout the execution of $\cal UBF$, augmenting paths from machines
in  ${\cal M}^{+}$ to ${\cal M}^{-}$ are found iteratively. Along each of these paths, the algorithm reassigns jobs between machines. Applying the algorithm results in reducing the makespan and balancing the loads.
The algorithm continues as long as there exists a path from ${\cal M}^{+}$ to ${\cal M}^{-}$.

\begin{figure}[th]
\begin{center}
{\epsfig{file=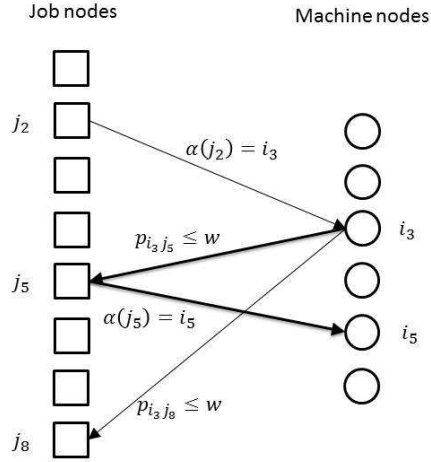,width=4.5in}}
\caption{\small The bipartite graph $G_\alpha (w)$. By changing the orientation of the path $i_{3}\rightarrow j_{5}\rightarrow i_{5}$, we remove job $j_{5}$ from $i_{5}$ and schedule it on machine $i_{3}$. }
\label{bipartiteGraph}
\end{center}
\end{figure}

%Next is a theorem similar to Theorem \ref{thm:approx_ratio} from the previous section.

\begin{theorem}
\label{thm:approx_ratio_restricted}
Let $\cal I$ be an instance of the restricted scheduling problem, with feasibility factor $\hepsilon \in \left(0,1\right] $.
 If there exists $0<q<\hepsilon $ such that $\frac{\sum_{j=1}^{n}p_{j}}{m}=q\cdot p_{max} $ then there 
exists a $\left(1+\frac{q}{\hepsilon} \right)$-approximation algorithm for the makespan, whose running time is 
$O(m^2n)$.
\end{theorem}

The correctness of Theorem \ref{thm:approx_ratio_restricted} follows from the next lemma.

\begin{lemma}
\label{lemma:UBF}
Let $\cal I$ be an instance of the restricted scheduling problem, with feasibility factor $\hepsilon \in \left(0,1\right] $. 
Then Algorithm $\cal UBF$ takes time $O(mS)$, where 
$S=
%\displaystyle{
\sum_{j=1}^{n}|\left\{ i:p_{ij}<\infty\right\} |
%}
$.
 Furthermore, For an initial assignment  $\alpha $,  $\cal UBF$$({\cal M},\alpha,\frac{L}{\hepsilon}, w)$
 for $w\geq p_{max}$ and $L=\frac {1}{m}\Sigma_{j=1}^{n}p_{j}$ terminates with ${\cal M}^{+}=\phi$.
\end{lemma}

\begin{proof}

Let $w\geq p_{max}$. We show that $\cal UBF$$({\cal M},\alpha,\frac{L}{\hepsilon}, w)$ terminates with $M^{+}=\phi$.
Let $\beta$ be the assignment computed by $\cal UBF$$({\cal M},\alpha,\frac{L}{\hepsilon}, w)$. By the observation, the average machine load of $\beta $ is $L$.
Assume that ${\cal M}^{+}\neq\emptyset$. Then there exists a machine $v$ with load $\delta_{v}({\cal I},\beta)> w+\frac{L}{\hepsilon}$. Since $\cal UBF$ terminated, we know that there is no path from a machine in ${\cal M}^{+}$ to a machine in ${\cal M}^{-}$ in the graph $G_\beta(w)$. Denote by ${\cal M}_v$ the set of machines reachable from $v$, ${\cal M}_{v}=\left\{ i\in M:\mbox{there is a directed path in \ensuremath{G_{\beta}(w)} from \ensuremath{v} to \ensuremath{i} }\right\}
 $.

Obviously, there is at least one job $u$ which is assigned to $v$ in $\beta$.
Thus, there is an edge of $\left(v,u\right)$ in $G_\beta(w)$. If $\hepsilon$ is the feasibility factor of $\cal I$, then there exists
at least $\hepsilon m$ machines that are good for $u$, which means that there exists an edge from $u$ to each of them. By
appending each of these edges to $\left(v,u\right)$ we get a directed path
from $v$ to at least $\hepsilon m$ machines (including $v$). Therefore
we can conclude that $\left|{\cal M}_v\right|\geq\hepsilon m$.

We compute now a lower bound of the average machine load  for $\beta$.
We sum the loads of all the machines
$i\in {\cal M}$. We know that $v\in {\cal M}^{+}$, thus $\delta_{v}({\cal I},\beta)> w+\frac{L}{\hepsilon}$. We also know
that there is no path from $v$ to machines in ${\cal M}^{-}$, and therefore
${\cal M}_v\cap {\cal M}^{-}=\emptyset$. Thus, for all $i\in {\cal M}_v$
$\delta_{i}({\cal I},\beta)>\frac{L}{\hepsilon}$ holds.

\[
\begin{array}{ll}
\sum_{i\in {\cal M}}\delta_{i}({\cal I},\beta) & \geq  \delta_{v}({\cal I},\beta)+\sum_{i\in {{\cal M}_v},i\neq v}\delta_{i}({\cal I},\beta)\\
 & > w+\frac{L}{\hepsilon}+\left(\left|{\cal M}_v\right|-1\right)\frac{L}{\hepsilon}\\
  & = w+\left|{\cal M}_v\right|\left(\frac{L}{\hepsilon}\right)\\
  & \geq w+mh\left(\frac{L}{\hepsilon}\right)\\
  & = w+mL\\
  & > mL \\
\end{array}
\]

We have shown that the sum of loads of the assignment $\beta$ is greater than $mL$. Hence, the average load for $\beta$ is greater then $L$ in contradiction to the average load of the resulting assignment $\beta$ being exactly $L$.
For the analysis of the running time, see \cite{GLMM04}.
\hfill \qed
\end{proof}

\noindent
{\bf Proof of Theorem \ref{thm:approx_ratio_restricted}.}
Let $\cal I$ be an instance of the restricted scheduling problem with feasibility factor $\hepsilon $.
Let $w=p_{max}$ and let $\alpha$ be some initial $w$-feasible assignment. The average machine load of $\alpha $ 
is $L=\frac {1}{m}\Sigma_{j=1}^{n}p_{j}$.
By Lemma \ref{lemma:UBF}, when $\cal UBF$$({\cal M},\alpha,\frac{L}{\hepsilon}, w)$
 terminates, we have that ${\cal M}^{+}=\emptyset $, i.e.,  
all machines are in ${\cal M}^{0}$ or in ${\cal M}^{-}$. 
Therefore, the maximum load of the resulting assignment is at most
 $w+\frac{L}{\hepsilon}=p_{max}+\frac{L}{\hepsilon}=p_{max}\left( 1+\frac{L}{\hepsilon p_{max}} \right)$.
Since the optimal makespan satisfies $OPT \geq p_{max}$, and $q=\frac{L}{ p_{max}}$, we have an approximation
ratio of $1+\frac{q}{\hepsilon}<2$. By Lemma \ref{lemma:UBF}, the running time of our algorithm is
$O(mS) = O(m^2 n)$.
\hfill  \qed

Note that our algorithm has better running time than the algorithm of \cite{GLMM04}, which 
uses binary search to find the best value for $\Delta$ yielding a $2$-approximation ratio to the minimum makespan.
%This binary search increases the overall running time 
This results in an overall running time of $O(mSlogW)$, where $W=\sum_{j=1}^{n}p_{j}$ is the sum of processing times of
all jobs.

\comment{
In the case where $\frac{L}{\epsilon}\geq T$ we just apply the same algorithm as in \cite{GLMM04} which guarantees an assignment with makespan at most $2T$. Therefore, our algorithm yields an assignment with makespan at most  $min\left\{ T+\frac{L}{\epsilon},2T\right\}$.
}

\comment{
\subsection{Uniform Restricted Instances}
In this section we consider the case where each job is feasible on exactly $k$ machines, for some $1\leq k\leq m$. If $\epsilon $ is the feasibility factor of such instances, then $k=\epsilon m$. We call this case \em{the uniform restricted case}.

\begin{lemma}
Let $\mathbf{P}$ be a uniform restricted instance and let $\epsilon$ be its feasibility factor. Then $\frac{L}{\epsilon}$ is a lower bound on the optimum makespan.
\end{lemma}

\begin{proof}

Partition the set of machines into $\frac{m}{k}$ sets $M_1,M_2,...M_{\frac{m}{k}}$  such that all the jobs
\hfill \qed
\end{proof}

\section{Discussion and Open Problems}
}

\medskip
\myparagraph{Acknowledgments} We thank Rohit Khandekar and Baruch Schieber for helpful discussions
on the paper.

\bibliographystyle{splncs}

\begin{thebibliography}{1}

\bibitem{ETKS08}
T. Ebenlendr, M. Kr\'{c}al, and J. Sgall.
Graph Balancing: a Special Case of Scheduling Unrelated
Parallel Machines. In {\em Proc. of the nineteenth annual
ACM-SIAM symposium on Discrete algorithms}, 483-490, 2008.


\bibitem{GLMM04}
M. Gairing, T. Lucking, M. Mavronicolas, and B. Monien.
Computing nash equilibria for scheduling on
restricted parallel links.
{\em Proceedings of the 36th Annual ACM Symposium
on the Thoery of Computing, STOC’04}, 613-622, 2004.


\bibitem{GMW07}
M. Gairing, B. Monien, and A. Woclaw.
A faster combinatorial approximation algorithm
for scheduling unrelated parallel machines.
{\em Theoretical Computer Science}, 380.1: 87-99, 2007.


\bibitem{H35}
P. Hall.
On Representatives of Subsets.
{\em J. London Math.} Soc 10.1: 26-30, 1935.


\bibitem{LST90}
 J. K. Lenstra, D. B. Shmoys, and É. Tardos.
Approximation algorithms for scheduling unrelated
parallel machines.
{\em Mathematical programming}, 46.1-3: 259-271, 1990.


\bibitem{SV05}
E. V. Shchepin, N. Vakhania.
An optimal rounding gives a better approximation for scheduling unrelated machines.
{\em Oper. Res. Lett.} 33(2): 127-133, 2005.

\bibitem{ST93}
D. B. Shmoys, and E. Tardos.
An approximation algorithm for the generalized assignment problem.
{\em Mathematical Programming}, 62.1-3: 461-474, 1993.


\bibitem{S12}
O.  Svensson. Santa Claus schedules
jobs on unrelated machines.
{\em SIAM Journal on Computing}, 41.5: 1318-1341, 2012.

\bibitem{T77}
Truemper, K. On max flows with gains and pure min-cost flows.
SIAM Journal on Applied Mathematics 32.2 (1977): 450-456.

\bibitem{V01}
V. Vazirani,
\newblock Approximation algorithms,
Springer Verlag, 2001.


\bibitem{JW11}
J. Verschae, and A.Wiese.
On the configuration-LP for scheduling on unrelated machines.
{\em Algorithms–ESA 2011. Springer Berlin Heidelberg}, 530-542, 2011.

\end{thebibliography}

%%%%%%%%%%%%%%%%%%%%%%%%%%%%%%%%%%%%%%%%%%%%%%%%%%%%%%%%%%%%%%%%%

%\appendix

\end{document}